\documentclass[11pt,onecolumn]{article}
\usepackage{graphics,graphicx,epsfig,amsmath,color}
\usepackage[T1]{fontenc}
\usepackage[left=2cm,right=2cm,top=2cm,bottom=2cm,bindingoffset=0cm]{geometry}
\usepackage{subfig}
\usepackage{amsmath, amsthm, amssymb}
\usepackage{graphicx}
\usepackage{authblk}

\usepackage[utf8]{inputenc}
\usepackage[english]{babel}
\usepackage{cite}
\usepackage{enumitem}
\usepackage{float}
\usepackage{color}
\usepackage{xcolor}
\usepackage{booktabs}
\usepackage{colortbl}
\usepackage{comment}
\usepackage[toc,page]{appendix}
\usepackage{tikz}
\usetikzlibrary{arrows.meta, bending, positioning, decorations.markings}

\tikzset{
    mid arrow/.style={
        postaction={
            decorate,
            decoration={
                markings,
                mark=at position 0.5 with {\arrow[#1, scale=2]{stealth}} 
            }
        }
    }
}

\newtheorem{proposition}{Proposition}
\newtheorem{theorem}{Theorem}
\newtheorem{corollary}{Corollary}

\begin{document}

\title{Heteroclinic networks in an ensemble of generalized Lotka-Volterra elements}
\author[1]{Alexander Korotkov}
\author[1]{Ekaterina Syundyukova}
\author[1]{Elena Gubina}
\author[1,2]{Grigory Osipov}
\affil[1]{Department of Control Theory and Dynamics of Systems, Lobachevsky State University of Nizhny Novgorod, 23, Gagarin Avenue, Nizhny Novgorod, 603022, Russia}
\affil[2]{Research and Education Mathematical Center "Mathematics for Future Technologies", Lobachevsky State University of Nizhny Novgorod, 23, Gagarin Avenue, Nizhny Novgorod, 603022, Russia}
\date{}
\maketitle

\begin{abstract}
In this article the generalized Lotka–Volterra model of ensemble of four excitory or inhibitory coupled elements are studied. It is shown that in the phase space of the model there exist heteroclinic network: a connected union of two or more heteroclinic cycles. A partition of the plane of coupling's parameters into sets of existence of various heteroclinic networks is constructed.
\end{abstract}

\section{Introduction}
Computational neuroscience has its roots in the fundamental research of Hodgkin and Huxley \cite{hodgkin1952currents, hodgkin1952components, hodgkin1952quantitative, hodgkin1952measurement}, which focused on the mechanisms of action potential propagation in the squid giant axon. Based on these works, a mathematical model of the neuron was developed, called the Hodgkin–Huxley model. This subsequently served as the basis for the creation of a number of simplified approaches to modeling neural activity. Neuron models can be divided into two types: microscopic-level (cell level) models and macroscopic-level (cell ensemble level) models. Models of the first type describe each neuron of the ensemble separately, its dynamics are defined by a system of differential or difference equations, and the connections between neurons are also formalized using the corresponding equations. The second type includes models that describe interactions between subpopulations of neurons, such as the Wilson–Cowan models \cite{wilson1972excitatory, wilson1973mathematical, destexhe2009wilson} and generalized Lotka–Volterra models \cite{fukai1997simple}. These models describe the interaction between subpopulations of neurons.


The generalized Lotka–Volterra models are widely used in natural and social sciences. For example, it can be used to describe models of population biology \cite{may1975nonlinear}, hydrodynamics \cite{busse1980convection, holmes2012turbulence}, economics \cite{orlando2021nonlinearities}, physics \cite{zhang2023emergent, roy2019numerical}.

The generalized Lotka–Volterra models can show WLC dynamics \cite{rabinovich2001dynamical, afraimovich2004heteroclinic, rabinovich2006dynamical, korotkov2025heteroclinic}. The concept of Winnerless Competition (WLC) is a general principle for information processing by dynamical systems. Winnerless competition is a concept that, when applied to neurodynamics, means that certain neurons or subpopulations of neurons temporarily become dominant while others are suppressed, so that the winner of the competition changes over time. The mathematical image of WLC dynamics is a stable heteroclinic cycle. A heteroclinic cycle can consist of heteroclinic trajectories and saddle equilibria or a saddle limit cycles. These saddle states corresponds to the activity of some group of neurons.

In addition to heteroclinic cycles, more complex heteroclinic structures can arise in the phase space of generalized Lotka–Volterra models. These are the so-called heteroclinic networks \cite{kirk1994competition, ashwin1999heteroclinic, schittler2012computation, koch2024biological}. Heteroclinic network is a connected union of two or more heteroclinic cycles. Several studies have shown that heteroclinical networks may underlie brain functions \cite{mazor2005transient, meyer2023heteroclinic, ashwin2024network, thakur2022heteroclinic, rossi2025dynamical}.

In this article the generalized Lotka–Volterra model of ensemble of four elements was studied. It is shown that in the phase space of four-dimensional model there exist heteroclinic network.

\section{Single element}
The dynamics of the activity of one isolated element is given by the equation
\begin{equation} \label{LV_equation}
\dot \rho = \rho(\gamma - \rho)(\rho - 1).
\end{equation}
Since the variable $\rho$ specifies the level of activity of the element, the constraints $0 \leq \rho \leq 1$ are imposed on it. From the equation \eqref{LV_equation} it is clear that the set $0 \leq \rho \leq 1$ is invariant. In this article we will assume that the parameter $\gamma$ obeys the constraints $0 < \gamma < 1$.

The system \eqref{LV_equation} has three equilibria: $E_1 = 0$, $E_2 = \gamma$ and $E_3 = 1$. Equilibria $E_1$ and
$E_3$ are stable, equilibrium $E_2$ is unstable. Thus, the parameter $\gamma$ can be interpreted as the excitation threshold of the element.

\section{The model of the ensemble of four elements}
The model of the ensemble of four elements has the form
\begin{equation} \label{LV_system_4}
\begin{cases}
\dot \rho_1 = \rho_1(\gamma - \rho_1 - \alpha \rho_2 - \beta \rho_3 - \alpha \rho_4)(\rho_1 - 1)\\
\dot \rho_2 = \rho_2(\gamma - \beta \rho_1 - \rho_2 - \alpha \rho_3 - \beta \rho_4)(\rho_2 - 1)\\
\dot \rho_3 = \rho_3(\gamma - \alpha \rho_1 - \beta \rho_2 - \rho_3 - \alpha \rho_4)(\rho_3 - 1)\\
\dot \rho_4 = \rho_4(\gamma - \beta \rho_1 - \alpha \rho_2 - \beta \rho_3 - \rho_4)(\rho_4 - 1)
\end{cases}.
\end{equation}
The hyperplanes $\rho_i = 0$ and $\rho_i = 1$ are invariant sets of this system. Therefore, the hypercube $0 \leq \rho_i \leq 1$ is also an invariant set of this system. In what follows, we will restrict our consideration to this set. In the case of positive parameters $\alpha$  and $\beta$ the coupling is inhibitory,  in the case of negative parameters $\alpha$  and $\beta$ the coupling is excitory.

The system \eqref{LV_system_4} is invariant with respect to the transformation of coordinates and parameters
\begin{equation} \label{sym}
T: (\rho_1, \rho_2, \rho_3, \rho_4, \alpha, \beta) \mapsto (\rho_4, \rho_3, \rho_2, \rho_1, \beta, \alpha).
\end{equation}

From the invariance with respect to the transformation $T$ it follows that the bifurcation diagram on the plane of parameters $(\alpha, \beta)$ is symmetric with respect to the diagonal $\beta = \alpha$, therefore we will further consider the case $\alpha \leq \beta$.

The system \eqref{LV_system_4} has $3^4 = 81$ equilibria. Some of them are presented in the table \ref{table1}.
\begin{table}[h]
    \centering
    \begin{tabular}{| c | c | c | c |}
        \hline
        Equilibrium & Coordinates & Equilibrium & Coordinates\\
        \hline
        $O$ & $(0, 0, 0, 0)$ & $O_{2, 3, 4}$ & $(0, 1, 1, 1)$\\
        \hline
        $O_1$ & $(1, 0, 0, 0)$ & $O_{1, 2, 3, 4}$ & $(1, 1, 1, 1)$\\
        \hline
        $O_2$ & $(0, 1, 0, 0)$ & $P_{1, 2}$ & $(1, \gamma - \beta, 0, 0)$\\
        \hline
        $O_3$ & $(0, 0, 1, 0)$ & $P_{1, 3}$ & $(1, 0, \gamma - \alpha, 0)$\\
        \hline
        $O_4$ & $(0, 0, 0, 1)$ & $P_{1, 4}$ & $(1, 0, 0, \gamma - \beta)$\\
        \hline
        $O_{1, 2}$ & $(1, 1, 0, 0)$ & $P_{2, 1}$ & $(\gamma - \alpha, 1, 0, 0)$\\
        \hline
        $O_{1, 3}$ & $(1, 0, 1, 0)$ & $P_{2, 3}$ & $(0, 1, \gamma - \beta, 0)$\\
        \hline
        $O_{1, 4}$ & $(1, 0, 0, 1)$ & $P_{2, 4}$ & $(0, 1, 0, \gamma - \alpha)$\\
        \hline
        $O_{2, 3}$ & $(0, 1, 1, 0)$ & $P_{3, 1}$ & $(\gamma - \beta, 0, 1, 0)$\\
        \hline
        $O_{2, 4}$ & $(0, 1, 0, 1)$ & $P_{3, 2}$ & $(0, \gamma - \alpha, 1, 0)$\\
        \hline
        $O_{3, 4}$ & $(0, 0, 1, 1)$ & $P_{3, 4}$ & $(0, 0, 1, \gamma - \beta)$\\
        \hline
        $O_{1, 2, 3}$ & $(1, 1, 1, 0)$ & $P_{4, 1}$ & $(\gamma - \alpha, 0, 0, 1)$\\
        \hline
        $O_{1, 2, 4}$ & $(1, 1, 0, 1)$ & $P_{4, 2}$ & $(0, \gamma - \beta, 0, 1)$\\
        \hline
        $O_{1, 3, 4}$ & $(1, 0, 1, 1)$ & $P_{4, 3}$ & $(0, 0, \gamma - \alpha, 1)$\\
        \hline
    \end{tabular}
    \caption{Some equilibria of the system \eqref{LV_system_4}.}
    \label{table1}
\end{table}

\subsection{Existence of heteroclinic cycles containing equilibria at the vertices of the invariant hypercube $0 \leq \rho_i \leq 1$}
Since all edges of the hypercube $0 \leq \rho_i \leq 1$ are invariant, then for the existence of a trajectory directed along the edge of the hypercube and connecting the equilibria $O_1$ and $O_{1, 2}$, it is necessary and sufficient that no other equilibria lie on this edge. On the line passing through the points $O_1$ and $O_{1, 2}$ there lies only one equilibrium $P_{1, 2}$. Therefore, the trajectory exists if one of the following inequalities is satisfied:
\begin{equation} \label{cond_1}
\left[
\begin{gathered}
\gamma - \beta > 1\\
\gamma - \beta < 0
\end{gathered}
\right..
\end{equation}
Also in this case the equilibria $P_{1, 4}$, $P_{2, 3}$, $P_{3, 1}$, $P_{3, 4}$, $P_{4, 2}$ do not lie on the segments connecting the equilibria $O_1$ and $O_{1, 4}$, $O_2$ and $O_{2, 3}$, $O_3$ and $O_{1, 3}$, $O_3$ and $O_{3, 4}$, $O_4$ and $O_{2, 4}$, respectively.

Similarly, if one of inequalities
\begin{equation} \label{cond_2}
\left[
\begin{gathered}
\gamma - \alpha > 1\\
\gamma - \alpha < 0
\end{gathered}
\right.
\end{equation}
is satisfied, then the equilibria $P_{1, 3}$, $P_{2, 1}$, $P_{2, 4}$, $P_{3, 2}$, $P_{4, 1}$, $P_{4, 3}$ do not lie on the segments between the equilibria $O_1$ and $O_{1, 3}$, $O_2$ and $O_{1, 2}$, $O_2$ and $O_{2, 4}$, $O_3$ and $O_{2, 3}$, $O_4$ and $O_{1, 4}$, $O_4$ and $O_{3, 4}$, respectively.

For the existence of trajectories between the equilibria $O_{1, 2}$ and $O_{1, 2, 3}$, $O_{1, 2}$ and $O_{1, 2, 4}$, $O_{1, 3}$ and $O_{1, 2, 3}$, $O_{2, 3}$ and $O_{1, 2, 3}$, $O_{2, 3}$ and $O_{2, 3, 4}$, $O_{2, 4}$ and $O_{2, 3, 4}$, $O_{3, 4}$ and $O_{1, 3, 4}$, $O_{3, 4}$ and $O_{2, 3, 4}$, it is necessary and sufficient that one of the following inequalities be satisfied:
\begin{equation} \label{cond_3}
\left[
\begin{gathered}
\gamma - \alpha - \beta > 1\\
\gamma - \alpha - \beta < 0
\end{gathered}
\right..
\end{equation}

Trajectories between the equilibria $O_{1, 3}$ and $O_{1, 3, 4}$, $O_{1, 4}$ and $O_{1, 2, 4}$ exist if one of the following inequalities is satisfied:
\begin{equation} \label{cond_4}
\left[
\begin{gathered}
\gamma - 2\beta > 1\\
\gamma - 2\beta < 0
\end{gathered}
\right..
\end{equation}

Trajectories between the equilibria $O_{1, 4}$ and $O_{1, 3, 4}$, $O_{2, 4}$ and $O_{1, 2, 4}$ exist if one of the following inequalities is satisfied:
\begin{equation} \label{cond_5}
\left[
\begin{gathered}
\gamma - 2\alpha > 1\\
\gamma - 2\alpha < 0
\end{gathered}
\right..
\end{equation}

Equilibria $O_{1, 2, 3}$ and $O_{1, 2, 3, 4}$, as well as $O_{1, 3, 4}$ and $O_{1, 2, 3, 4}$ are connected by trajectories if one of the following inequalities is satisfied:
\begin{equation} \label{cond_6}
\left[
\begin{gathered}
\gamma - \alpha - 2\beta > 1\\
\gamma - \alpha - 2\beta < 0
\end{gathered}
\right..
\end{equation}

Equilibria $O_{1, 2, 4}$ and $O_{1, 2, 3, 4}$, as well as $O_{2, 3, 4}$ and $O_{1, 2, 3, 4}$ are connected by trajectories if one of the following inequalities is satisfied:
\begin{equation} \label{cond_7}
\left[
\begin{gathered}
\gamma - 2\alpha - \beta > 1\\
\gamma - 2\alpha - \beta < 0
\end{gathered}
\right..
\end{equation}

Figure \ref{graph4} shows a graph of possible transitions between equilibria at the vertices of the hypercube $0 \leq \rho_i \leq 1$. Near each edge of the graph there is a label corresponding to one of the conditions \eqref{cond_1}-\eqref{cond_7} (label $E1$ corresponds to condition \eqref{cond_1}, label $E2$ corresponds to condition \eqref{cond_2}, etc). For each edge of the graph to exist, the condition specified next to that edge must be satisfied. If the first inequality from the corresponding set of inequalities is satisfied, then movement along the edge occurs upwards, otherwise, downwards.
\begin{figure}[H]
    \centering
    \includegraphics[width = 1\linewidth] {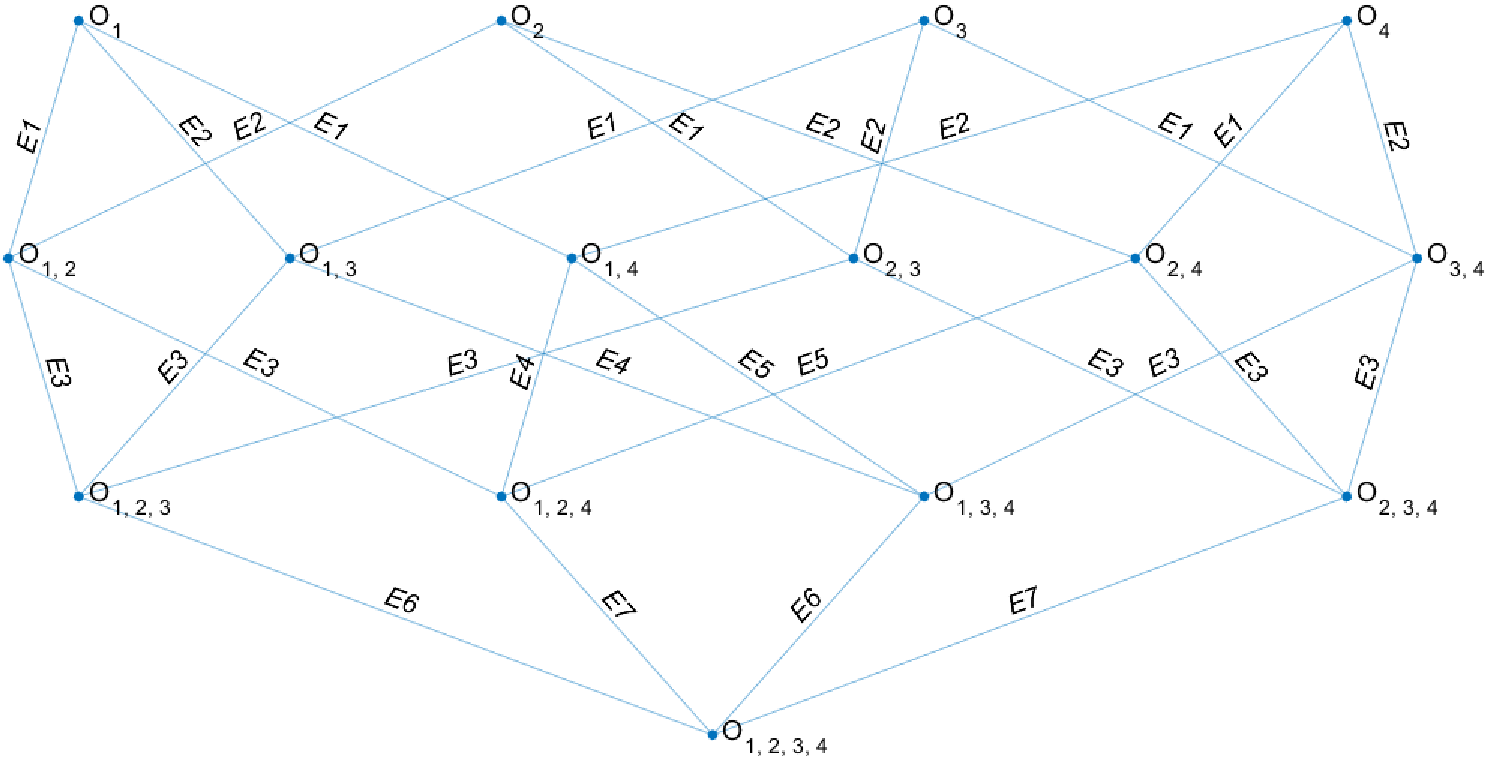}
    \caption{Graph of possible transitions between equilibria at the vertices of the invariant hypercube $0 \leq \rho_i \leq 1$.}
    \label{graph4}
\end{figure}

Since the graph \ref{graph4} does not contain edges directed horizontally, all possible cycles in it have even length. Proposition \ref{prop1} states that all heteroclinic cycles have at least six equilibria (see appendix).

Note that from the condition $0 < \gamma < 1$ it follows that the inequality $\gamma - \alpha > 1$ implies the inequality $\gamma - 2\alpha > 1$, and the inequality $\gamma - \alpha < 0$ implies the inequality $\gamma - 2\alpha < 0$. It follows that the existence of an edge with label $E2$ implies the existence of an edge with label $E5$, and if both of these edges exist, then they have the same direction. Similarly, the existence of an edge with label $E1$ implies the existence of an edge with label $E4$, and if both of these edges exist, then they have the same direction.

In the phase space of the system \eqref{LV_system_4} there exist $21$ heteroclinic cycles $\Gamma_i$, containing only those equilibria that lie at the vertices of the hypercube $0 \leq \rho_i \leq 1$ (see appendix).

Among the heteroclinic cycles given, four cycles contain six equilibria (cycles $\Gamma_1$, $\Gamma_{16}$, $\Gamma_{20}$, $\Gamma_{21}$), fifteen cycles contain eight equilibria (cycles $\Gamma_3$, $\Gamma_4$, $\Gamma_5$, $\Gamma_6$, $\Gamma_7$, $\Gamma_8$, $\Gamma_9$, $\Gamma_{10}$, $\Gamma_{11}$, $\Gamma_{12}$, $\Gamma_{13}$, $\Gamma_{14}$, $\Gamma_{17}$, $\Gamma_{18}$, $\Gamma_{19}$) and two cycles contain ten equilibria (cycles $\Gamma_2$ and $\Gamma_{15}$).

The transformation \eqref{sym} assigns to each of the found cycles a cycle obtained by renumbering the equilibria included in this cycle according to the transformation \eqref{sym}. The set of existence of the cycle in parameter plane $(\alpha, \beta)$ is also transformed according to this law. Therefore, the existence sets of the corresponding cycles are symmetric with respect to the line $\alpha = \beta$.

In addition to the \eqref{sym} transformation, one more transformation can be applied to the found heteroclinic cycles: to do this, it is necessary to change the direction of movement along the cycle to the opposite. In this case, the existence set of the resulting cycle is defined as follows: each inequality defining the set of existence of the original cycle must be replaced by one that is included together with it in one of the sets \eqref{cond_1}-\eqref{cond_7}.

For example, the cycle $\Gamma_{21}: O_3 \rightarrow O_{1, 3} \rightarrow O_{1, 3, 4} \rightarrow O_{1, 2, 3, 4} \rightarrow O_{2, 3, 4} \rightarrow O_{2, 3} \rightarrow O_3$ is mapped by \eqref{sym} to the cycle $\Gamma'_{21}: O_2 \rightarrow O_{2, 4} \rightarrow O_{1, 2, 4} \rightarrow O_{1, 2, 3, 4} \rightarrow O_{1, 2, 3} \rightarrow O_{2, 3} \rightarrow O_2$. The existence set of the cycle $\Gamma'_{21}$ is determined by the following system of inequalities:
\begin{equation} \label{example}
\begin{cases}
\gamma - \alpha - \beta > 1\\
\gamma - \alpha - 2\beta > 1\\
\gamma - \beta > 1\\
\gamma - 2\alpha - \beta < 0\\
\gamma - \alpha < 0\\
\gamma - 2\alpha < 0
\end{cases}.
\end{equation}
This condition is obtained from the condition defining the existence set of the cycle $\Gamma_{21}$ by replacing $\alpha \leftrightarrow \beta$.

The second transformation $R$ maps the cycle $\Gamma'_{21}$ to the cycle $\Gamma_{20}: O_2 \rightarrow O_{2, 3} \rightarrow O_{1, 2, 3} \rightarrow O_{1, 2, 3, 4} \rightarrow O_{1, 2, 4} \rightarrow O_{2, 4} \rightarrow O_2$. Its existence set is obtained from the system \eqref{example} by replacing each inequality of the form $\ldots < 1$ with an inequality of the form $\ldots > 0$ with the same left-hand side and vice versa. As a result, we obtain a system that defines the existence set of the cycle $\Gamma_{20}$.

Proposition \ref{prop2} (see appendix) states that all heteroclinic cycles contain at least one of $O_i$ equilibria. Corollary \ref{prop3} (see appendix) states that necessary condition for the existence of a heteroclinic cycle with equilibria at the vertices of the hypercube $0 \leq \rho_i \leq 1$ is the fulfillment of the inequality $\alpha + 1 < \gamma < \beta$ or $\beta + 1 < \gamma < \alpha$.
As shown in appendix, for the existence of cycles from the set $\{\Gamma_1, \Gamma_5, \Gamma_6, \Gamma_{10}, \Gamma_{11}, \Gamma_{16}\}$ it is sufficient that the inequalities $\alpha + 1 < \gamma < \beta$ are satisfied. Thus, we obtained the theorem.
\begin{theorem}
For the existence of a heteroclinic cycle with equilibria at the vertices of the hypercube $0 \leq \rho_i \leq 1$, it is necessary and sufficient that the inequalities $\alpha + 1 < \gamma < \beta$ or $\beta + 1 < \gamma < \alpha$ are satisfied.
\end{theorem}

\subsection{Heteroclinic networks}
A heteroclinic network is a finite set of saddle equilibria and separatrices connecting them, such that for any two equilibria $Q_1$, $Q_2$ of this network, one can move from $Q_1$ to $Q_2$ along heteroclinic trajectories of this network.

Examples of heteroclinic networks:

\begin{center}
\begin{tikzpicture}
\draw[thick, mid arrow] (2, 4) -- (0, 2) node[anchor=east] {$O_1$};
\draw[thick, mid arrow] (0, 2) -- (4, 2) node[anchor=west] {$O_2$};
\draw[thick, mid arrow] (4, 2) -- (2, 4) node[anchor=west] {$O_3$};
\draw[thick, mid arrow] (4, 2) -- (2, 0) node[anchor=west] {$O_4$};
\draw[thick, mid arrow] (2, 0) -- (0, 2);
\end{tikzpicture}
\end{center}

\begin{center}
\resizebox{10.3cm}{!}{
\begin{tikzpicture}
\draw[thick, mid arrow] (2, 4) -- (0, 2) node[anchor=east] (1) {$O_1$};
\draw[thick, mid arrow] (0, 2) -- (4, 2) node[anchor=south west] {$O_2$};
\draw[thick, mid arrow] (4, 2) -- (2, 4) node[anchor=west] (3) {$O_3$};
\draw[thick, mid arrow] (8, 4) -- (6, 2) node[anchor=south east] {$O_4$};
\draw[thick, mid arrow] (6, 2) -- (10, 2) node[anchor=west] (5) {$O_5$};
\draw[thick, mid arrow] (10, 2) -- (8, 4) node[anchor=east] (6) {$O_6$};
\draw[thick, mid arrow] (6, 2) -- (4, 2);
\draw (3) edge [thick, mid arrow, bend left=30] (6);
\end{tikzpicture}}
\end{center}

An example of a set of saddle equilibria and heteroclinic trajectories that do not form a heteroclinic network:

\begin{center}
\resizebox{10.5cm}{!}{
\begin{tikzpicture}
\draw[thick, mid arrow] (2, 4) -- (0, 2) node[anchor=east] {$O_1$};
\draw[thick, mid arrow] (0, 2) -- (4, 2) node[anchor=south west] {$O_2$};
\draw[thick, mid arrow] (4, 2) -- (2, 4) node[anchor=west] {$O_3$};
\draw[thick, mid arrow] (8, 4) -- (6, 2) node[anchor=south east] {$O_4$};
\draw[thick, mid arrow] (6, 2) -- (10, 2) node[anchor=west] {$O_5$};
\draw[thick, mid arrow] (10, 2) -- (8, 4) node[anchor=west] {$O_6$};
\draw[thick, mid arrow] (6, 2) -- (4, 2);
\end{tikzpicture}}
\end{center}

From the inequalities \eqref{cond_hc_1}-\eqref{cond_hc_6} it follows that in the phase space of the system \eqref{LV_system_4}, depending on the values of the parameters, the following heteroclinic networks can exist:
\begin{enumerate}
    \item In the set $0 < \gamma - \alpha - \beta < 1$, $\gamma - \alpha - 2\beta < 0$, $\gamma - \beta < 0$, $\gamma - 2\alpha - \beta > 1$, the conditions \eqref{cond_hc_1} and \eqref{cond_hc_6} are satisfied, therefore, there exists a heteroclinic network $\frak{G_1}$ consisting of cycles $$\{\Gamma_1, \Gamma_5, \Gamma_6, \Gamma_{10}, \Gamma_{11}, \Gamma_{13}, \Gamma_{16}\}.$$

    \item In the set $\gamma - \alpha - \beta < 0$, $\gamma - 2\alpha - \beta > 1$ the conditions \eqref{cond_hc_1}, \eqref{cond_hc_2}, \eqref{cond_hc_3} and \eqref{cond_hc_6} are satisfied, therefore, there exists a heteroclinic network $\frak{G_2}$ consisting of cycles $$\{\Gamma_1, \Gamma_2, \Gamma_3, \Gamma_5, \Gamma_6, \Gamma_9, \Gamma_{10}, \Gamma_{11}, \Gamma_{13}, \Gamma_{16}, \Gamma_{17}, \Gamma_{18}, \Gamma_{19}, \Gamma_{20}\}.$$

    \item In the set $\gamma - \alpha - \beta > 1$, $\gamma - \alpha - 2\beta < 0$ the conditions \eqref{cond_hc_1}, \eqref{cond_hc_4}, \eqref{cond_hc_5} and \eqref{cond_hc_6} are satisfied, therefore, there exists a heteroclinic network $\frak{G_3}$ consisting of cycles $$\{\Gamma_1, \Gamma_4, \Gamma_5, \Gamma_6, \Gamma_7, \Gamma_8, \Gamma_{10}, \Gamma_{11}, \Gamma_{12}, \Gamma_{13}, \Gamma_{14}, \Gamma_{15}, \Gamma_{16}, \Gamma_{21}\}.$$

    \item In the set $\gamma - \alpha - \beta > 1$, $\gamma - \alpha - 2\beta > 0$, $\gamma - \beta < 0$ the conditions \eqref{cond_hc_1} and \eqref{cond_hc_4} are satisfied, therefore, there exists a heteroclinic network $\frak{G_4}$ consisting of cycles $$\{\Gamma_1, \Gamma_4, \Gamma_5, \Gamma_6, \Gamma_7, \Gamma_{10}, \Gamma_{11}, \Gamma_{12}, \Gamma_{16}\}.$$

    \item In the set $\gamma - 2\alpha - \beta < 1$, $\gamma - \alpha - \beta < 0$, $\gamma - \alpha > 1$ the conditions \eqref{cond_hc_1} and \eqref{cond_hc_3} are satisfied, therefore, there exists a heteroclinic network $\frak{G_5}$ consisting of cycles $$\{\Gamma_1, \Gamma_3, \Gamma_5, \Gamma_6, \Gamma_{10}, \Gamma_{11}, \Gamma_{16}, \Gamma_{17}, \Gamma_{18}\}.$$

    \item In the intersection of the set $\gamma - \beta < 0$, $\gamma - \alpha > 1$ with the complement to the union of all sets from the previous points, only the condition \eqref{cond_hc_1} is satisfied, therefore, there is a heteroclinic network $\frak{G_6}$ consisting of cycles $$\{\Gamma_1, \Gamma_5, \Gamma_6, \Gamma_{10}, \Gamma_{11}, \Gamma_{16}\}.$$
\end{enumerate}

The figure \ref{partition} shows the partition of the set $\alpha + 1 < \gamma < \beta$ into the existence sets of the specified heteroclinic networks.
\begin{figure}[H]
    \centering
    \includegraphics[width = .4\linewidth]{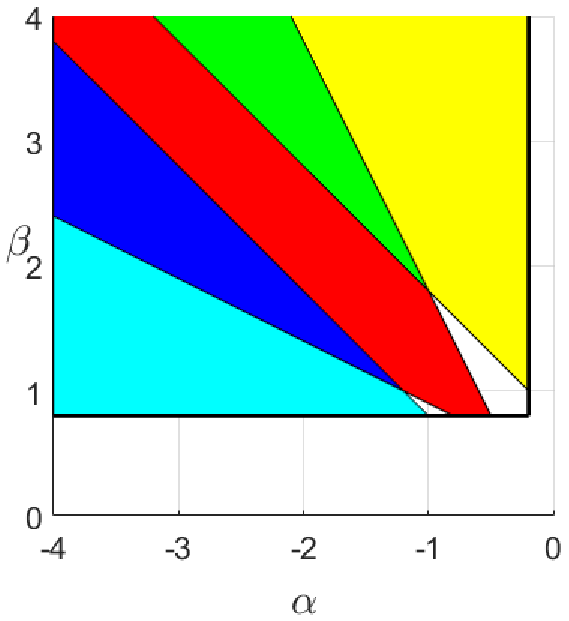}
    \caption{Partition of the set $\alpha + 1 < \gamma < \beta$ (the boundaries of the set are thick black lines) into sets of existence of different heteroclinic networks. The set marked in red is the existence set of the network $\frak{G_1}$, the set marked in green is the existence set of the network $\frak{G_2}$, the set marked in blue is the existence set of the network $\frak{G_3}$, the set marked in cyan is the existence set of the network $\frak{G_4}$, the set marked in yellow is the existence set of the network $\frak{G_5}$, the set marked in white is the existence set of the network $\frak{G_6}$.}
    \label{partition}
\end{figure}

\begin{figure}[H]
    \centering
    \includegraphics[width = 1\linewidth] {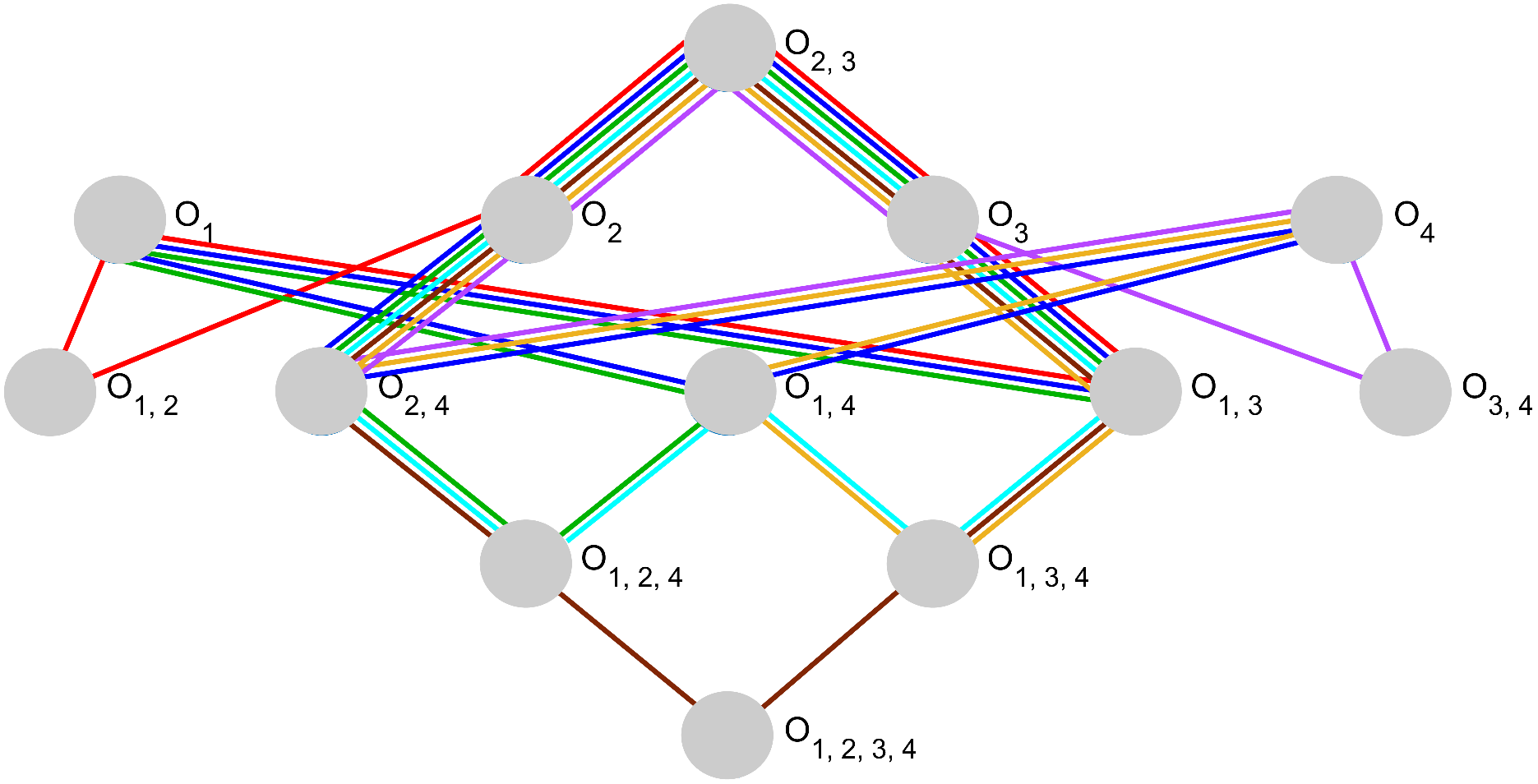}
    \caption{Graph showing a heteroclinic network consisting of cycles $\Gamma_1$ (is shown in red), $\Gamma_5$ (is shown in blue), $\Gamma_6$ (is shown in green), $\Gamma_{10}$ (is shown in light brown), $\Gamma_{11}$ (is shown in cyan), $\Gamma_{13}$ (is shown in brown), $\Gamma_{16}$ (is shown in violet). Each of the network cycles has a symmetric cycle in the same network with respect to the line passing through the vertices $O_{2, 3}$, $O_{1, 4}$ and $O_{1, 2, 3, 4}$. The cycle $\Gamma_1$ is symmetric to the cycle $\Gamma_{16}$, the cycle $\Gamma_6$ is symmetric to the cycle $\Gamma_{10}$, and the remaining cycles are symmetric to themselves.}
    \label{graph_for_network}
\end{figure}
Figure \ref{graph_for_network} shows a graph in which the cycles included in the heteroclinic network $\{\Gamma_1, \Gamma_5, \Gamma_6, \Gamma_{10}, \Gamma_{11}, \Gamma_{13}, \Gamma_{16}\}$ are marked.

To obtain a trajectory tending to a stable heteroclinic cycle, we make the substitutions $x_i = \ln \frac{\rho_i}{1 - \rho_i}$. As a result, the system \eqref{LV_system_4} takes the following form:
\begin{equation}
\begin{cases}
\dot x_1 = \frac{1}{1 + e^{-x_1}} + \frac{\alpha}{1 + e^{-x_2}} + \frac{\beta}{1 + e^{-x_3}} + \frac{\alpha}{1 + e^{-x_4}} - \gamma\\
\dot x_2 = \frac{\beta}{1 + e^{-x_1}} + \frac{1}{1 + e^{-x_2}} + \frac{\alpha}{1 + e^{-x_3}} + \frac{\beta}{1 + e^{-x_4}} - \gamma\\
\dot x_3 = \frac{\alpha}{1 + e^{-x_1}} + \frac{\beta}{1 + e^{-x_2}} + \frac{1}{1 + e^{-x_3}} + \frac{\alpha}{1 + e^{-x_4}} - \gamma\\
\dot x_4 = \frac{\beta}{1 + e^{-x_1}} + \frac{\alpha}{1 + e^{-x_2}} + \frac{\beta}{1 + e^{-x_3}} + \frac{1}{1 + e^{-x_4}} - \gamma
\end{cases}.
\end{equation}
The corresponding time series are shown in figure \ref{time_series}.

\begin{figure}[H]
    \centering
    \subfloat[]
    {
        \includegraphics[width = .45\linewidth] {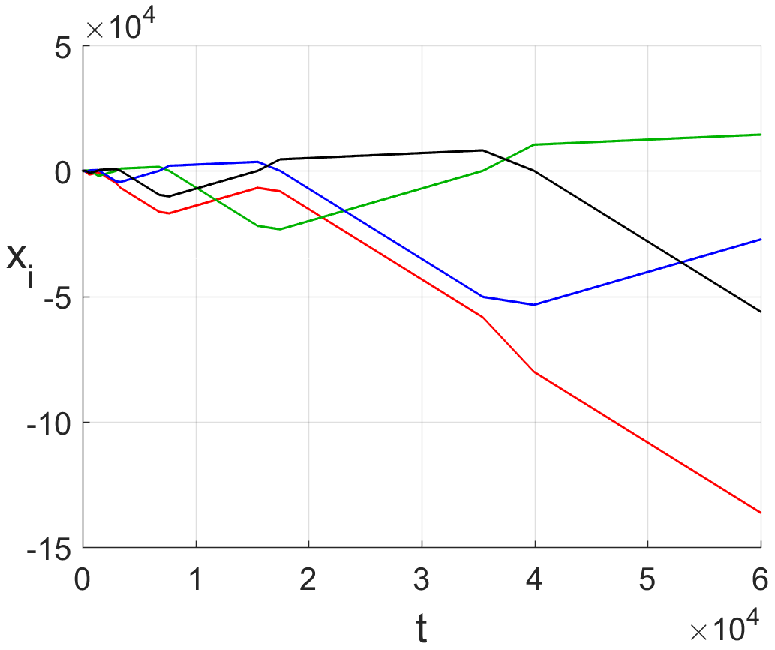}
    }    
    \subfloat[]
    {
        \includegraphics[width = .45\linewidth] {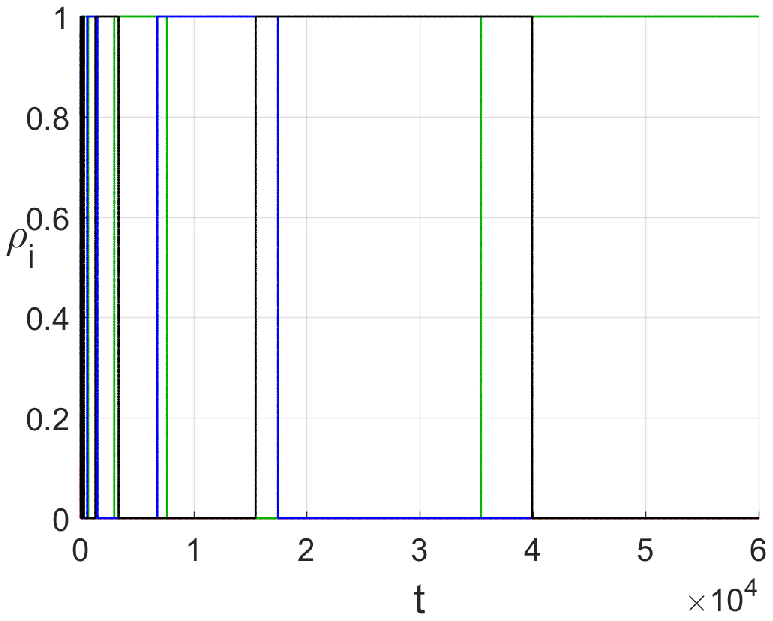}
    }
    \caption{Time series (in the figure (a) new coordinates $x_i$ are used, in the figure (b) original coordinates $\rho_i$ are used). Initial conditions: (0.65, 0.81, 0.67, 0.97). Parameters: $\alpha = -2$, $\beta = 2.1$, $\gamma = 0.8$.}
    \label{time_series}
\end{figure}

\section{Conclusions}
\begin{itemize}
    \item The article proposes and studies a model of an ensemble of four generalized Lotka-Volterra elements, in which the couplings between the elements can be excitatory, inhibitory or mixed (both excitatory and inhibitory) types depending on the values of the parameters.
    \item Necessary and sufficient conditions for the existence of heteroclinic cycles lying on the edges of the invariant hypercube are found.
    \item It is shown existence of heteroclinic networks. A partition of the parameter space into sets of existence of various heteroclinic networks is given.
    \item Due to the presence in the phase space of the proposed model of a stable heteroclinic cycle, depending on the initial conditions, the following is possible: a) absence of neuron-like activity (mathematical image is stable equilibrium in origin); b) non-decaying sequential switching activity of single neuron, then couple of neurons, then another single neuron, then couple of other neurons, etc. (mathematical image is stable heteroclinic cycle).
\end{itemize}

\section{Acknowledgments}
The authors thank S. V. Gonchenko for helpful advices.
The study is supported by the Ministry of Science and Education of the Russian Federation, Agreement \#FSWR-2020-0036.

\section*{Data Availability Statement}
The data that support the findings of this study are available within the article [and its supplementary material].

\appendix
\section{Heteroclinic cycles}
In the phase space of the system \eqref{LV_system_4} there exist the following heteroclinic cycles, containing only those equilibria that lie at the vertices of the hypercube $0 \leq \rho_i \leq 1$:
\begin{enumerate}
    \item The cycle $\Gamma_1: O_1 \rightarrow O_{1, 2} \rightarrow O_2 \rightarrow O_{2, 3} \rightarrow O_3 \rightarrow O_{1, 3} \rightarrow O_1$ exists when the following condition is satisfied: $\alpha + 1 < \gamma < \beta$.
    
    \item The cycle $\Gamma_2: O_1 \rightarrow O_{1, 2} \rightarrow O_{1, 2, 3} \rightarrow O_{1, 2, 3, 4} \rightarrow O_{1, 2, 4} \rightarrow O_{2, 4} \rightarrow O_2 \rightarrow O_{2, 3} \rightarrow O_3 \rightarrow O_{1, 3} \rightarrow O_1$ exists when the following condition is satisfied:
    \begin{equation*}
    \begin{cases}
    \gamma - 2\alpha - \beta > 1\\
    \gamma - \alpha > 1\\
    \gamma - 2\alpha > 1\\
    \gamma - \alpha - \beta < 0\\
    \gamma - \alpha - 2\beta < 0\\
    \gamma - \beta < 0
    \end{cases}
    \Longleftrightarrow
    \begin{cases}
    \gamma - 2\alpha - \beta > 1\\
    \gamma - \alpha - \beta < 0
    \end{cases}.
    \end{equation*}
    
    \item The cycle $\Gamma_3: O_1 \rightarrow O_{1, 2} \rightarrow O_{1, 2, 4} \rightarrow O_{2, 4} \rightarrow O_2 \rightarrow O_{2, 3} \rightarrow O_3 \rightarrow O_{1, 3} \rightarrow O_1$ exists when the following condition is satisfied:
    \begin{equation*}
    \begin{cases}
    \gamma - \alpha - \beta < 0\\
    \gamma - \alpha > 1\\
    \gamma - 2\alpha > 1\\
    \gamma - \beta < 0
    \end{cases}
    \Longleftrightarrow
    \begin{cases}
    \gamma - \alpha - \beta < 0\\
    \gamma - \alpha > 1
    \end{cases}.
    \end{equation*}
    
    \item The cycle $\Gamma_4: O_1 \rightarrow O_{1, 4} \rightarrow O_{1, 2, 4} \rightarrow O_{1, 2} \rightarrow O_2 \rightarrow O_{2, 3} \rightarrow O_3 \rightarrow O_{1, 3} \rightarrow O_1$ exists when the following condition is satisfied:
    \begin{equation*}
    \begin{cases}
    \gamma - \alpha - \beta > 1\\
    \gamma - \alpha > 1\\
    \gamma - \beta < 0\\
    \gamma - 2\beta < 0
    \end{cases}
    \Longleftrightarrow
    \begin{cases}
    \gamma - \alpha - \beta > 1\\
    \gamma - \beta < 0
    \end{cases}.
    \end{equation*}
    
    \item The cycle $\Gamma_5: O_1 \rightarrow O_{1, 4} \rightarrow O_4 \rightarrow O_{2, 4} \rightarrow O_2 \rightarrow O_{2, 3} \rightarrow O_3 \rightarrow O_{1, 3} \rightarrow O_1$ exists when the following condition is satisfied: $\alpha + 1 < \gamma < \beta$.
    
    \item The cycle $\Gamma_6: O_1 \rightarrow O_{1, 4} \rightarrow O_{1, 2, 4} \rightarrow O_{2, 4} \rightarrow O_2 \rightarrow O_{2, 3} \rightarrow O_3 \rightarrow O_{1, 3} \rightarrow O_1$ exists when the following condition is satisfied:
    \begin{equation*}
    \begin{cases}
    \gamma - \alpha > 1\\
    \gamma - 2\alpha > 1\\
    \gamma - \beta < 0\\
    \gamma - 2\beta < 0
    \end{cases}
    \Longleftrightarrow
    \alpha + 1 < \gamma < \beta.
    \end{equation*}
    
    \item The cycle $\Gamma_7: O_2 \rightarrow O_{2, 3} \rightarrow O_3 \rightarrow O_{1, 3} \rightarrow O_{1, 3, 4} \rightarrow O_{1, 4} \rightarrow O_{1, 2, 4} \rightarrow O_{1, 2} \rightarrow O_2$ exists when the following condition is satisfied:
    \begin{equation*}
    \begin{cases}
    \gamma - \alpha - \beta > 1\\
    \gamma - \alpha > 1\\
    \gamma - 2\alpha > 1\\
    \gamma - \beta < 0\\
    \gamma - 2\beta < 0
    \end{cases}
    \Longleftrightarrow
    \begin{cases}
    \gamma - \alpha - \beta > 1\\
    \gamma - \beta < 0
    \end{cases}.
    \end{equation*}
    
    \item The cycle $\Gamma_8: O_2 \rightarrow O_{2, 3} \rightarrow O_3 \rightarrow O_{1, 3} \rightarrow O_{1, 3, 4} \rightarrow O_{1, 2, 3, 4} \rightarrow O_{1, 2, 4} \rightarrow O_{1, 2} \rightarrow O_2$ exists when the following condition is satisfied:
    \begin{equation*}
    \begin{cases}
    \gamma - \alpha - \beta > 1\\
    \gamma - 2\alpha - \beta > 1\\
    \gamma - \alpha > 1\\
    \gamma - \alpha - 2\beta < 0\\
    \gamma - \beta < 0\\
    \gamma - 2\beta < 0
    \end{cases}
    \Longleftrightarrow
    \begin{cases}
    \gamma - \alpha - \beta > 1\\
    \gamma - \alpha - 2\beta < 0
    \end{cases}.
    \end{equation*}
    
    \item The cycle $\Gamma_9: O_2 \rightarrow O_{2, 3} \rightarrow O_3 \rightarrow O_{1, 3} \rightarrow O_{1, 2, 3} \rightarrow O_{1, 2, 3, 4} \rightarrow O_{1, 2, 4} \rightarrow O_{2, 4} \rightarrow O_2$ exists when the following condition is satisfied:
    \begin{equation*}
    \begin{cases}
    \gamma - 2\alpha - \beta > 1\\
    \gamma - 2\alpha > 1\\
    \gamma - \alpha > 1\\
    \gamma - \alpha - 2\beta < 0\\
    \gamma - \alpha - \beta < 0\\
    \gamma - \beta < 0
    \end{cases}
    \Longleftrightarrow
    \begin{cases}
    \gamma - 2\alpha - \beta > 1\\
    \gamma - \alpha - \beta < 0
    \end{cases}.
    \end{equation*}
    
    \item The cycle $\Gamma_{10}: O_2 \rightarrow O_{2, 3} \rightarrow O_3 \rightarrow O_{1, 3} \rightarrow O_{1, 3, 4} \rightarrow O_{1, 4} \rightarrow O_4 \rightarrow O_{2, 4} \rightarrow O_2$ exists when the following condition is satisfied:
    \begin{equation*}
    \begin{cases}
    \gamma - \alpha > 1\\
    \gamma - 2\alpha > 1\\
    \gamma - \beta < 0\\
    \gamma - 2\beta < 0
    \end{cases}
    \Longleftrightarrow
    \alpha + 1 < \gamma < \beta.
    \end{equation*}
    
    \item The cycle $\Gamma_{11}: O_2 \rightarrow O_{2, 3} \rightarrow O_3 \rightarrow O_{1, 3} \rightarrow O_{1, 3, 4} \rightarrow O_{1, 4} \rightarrow O_{1, 2, 4} \rightarrow O_{2, 4} \rightarrow O_2$ exists when the following condition is satisfied:
    \begin{equation*}
    \begin{cases}
    \gamma - \alpha > 1\\
    \gamma - 2\alpha > 1\\
    \gamma - \beta < 0\\
    \gamma - 2\beta < 0
    \end{cases}
    \Longleftrightarrow
    \alpha + 1 < \gamma < \beta.
    \end{equation*}
    
    \item The cycle $\Gamma_{12}: O_2 \rightarrow O_{2, 3} \rightarrow O_3 \rightarrow O_{1, 3} \rightarrow O_{1, 3, 4} \rightarrow O_{3, 4} \rightarrow O_4 \rightarrow O_{2, 4} \rightarrow O_2$ exists when the following condition is satisfied:
    \begin{equation*}
    \begin{cases}
    \gamma - \alpha - \beta > 1\\
    \gamma - \alpha > 1\\
    \gamma - \beta < 0\\
    \gamma - 2\beta < 0
    \end{cases}
    \Longleftrightarrow
    \begin{cases}
    \gamma - \alpha - \beta > 1\\
    \gamma - \beta < 0
    \end{cases}.
    \end{equation*}
    
    \item The cycle $\Gamma_{13}: O_2 \rightarrow O_{2, 3} \rightarrow O_3 \rightarrow O_{1, 3} \rightarrow O_{1, 3, 4} \rightarrow O_{1, 2, 3, 4} \rightarrow O_{1, 2, 4} \rightarrow O_{2, 4} \rightarrow O_2$ exists when the following condition is satisfied:
    \begin{equation*}
    \begin{cases}
    \gamma - 2\alpha - \beta > 1\\
    \gamma - \alpha > 1\\
    \gamma - 2\alpha > 1\\
    \gamma - \alpha - 2\beta < 0\\
    \gamma - \beta < 0\\
    \gamma - 2\beta < 0
    \end{cases}
    \Longleftrightarrow
    \begin{cases}
    \gamma - 2\alpha - \beta > 1\\
    \gamma - \alpha - 2\beta < 0\\
    \gamma - \beta < 0
    \end{cases}.
    \end{equation*}
    
    \item The cycle $\Gamma_{14}: O_2 \rightarrow O_{2, 3} \rightarrow O_3 \rightarrow O_{1, 3} \rightarrow O_{1, 3, 4} \rightarrow O_{1, 2, 3, 4} \rightarrow O_{2, 3, 4} \rightarrow O_{2, 4} \rightarrow O_2$ exists when the following condition is satisfied:
    \begin{equation*}
    \begin{cases}
    \gamma - \alpha - \beta > 1\\
    \gamma - 2\alpha - \beta > 1\\
    \gamma - \alpha > 1\\
    \gamma - \alpha - 2\beta < 0\\
    \gamma - \beta < 0\\
    \gamma - 2\beta < 0
    \end{cases}
    \Longleftrightarrow
    \begin{cases}
    \gamma - \alpha - \beta > 1\\
    \gamma - \alpha - 2\beta < 0
    \end{cases}.
    \end{equation*}
    
    \item The cycle $\Gamma_{15}: O_2 \rightarrow O_{2, 3} \rightarrow O_3 \rightarrow O_{1, 3} \rightarrow O_{1, 3, 4} \rightarrow O_{1, 2, 3, 4} \rightarrow O_{2, 3, 4} \rightarrow O_{3, 4} \rightarrow O_4 \rightarrow O_{2, 4} \rightarrow O_2$ exists when the following condition is satisfied:
    \begin{equation*}
    \begin{cases}
    \gamma - \alpha - \beta > 1\\
    \gamma - 2\alpha - \beta > 1\\
    \gamma - \alpha > 1\\
    \gamma - \alpha - 2\beta < 0\\
    \gamma - \beta < 0\\
    \gamma - 2\beta < 0
    \end{cases}
    \Longleftrightarrow
    \begin{cases}
    \gamma - \alpha - \beta > 1\\
    \gamma - \alpha - 2\beta < 0
    \end{cases}.
    \end{equation*}
    
    \item The cycle $\Gamma_{16}: O_2 \rightarrow O_{2, 3} \rightarrow O_3 \rightarrow O_{3, 4} \rightarrow O_4 \rightarrow O_{2, 4} \rightarrow O_2$ exists when the following condition is satisfied: $\alpha + 1 < \gamma < \beta$.
    
    \item The cycle $\Gamma_{17}: O_2 \rightarrow O_{2, 3} \rightarrow O_3 \rightarrow O_{3, 4} \rightarrow O_{1, 3, 4} \rightarrow O_{1, 4} \rightarrow O_4 \rightarrow O_{2, 4} \rightarrow O_2$ exists when the following condition is satisfied:
    \begin{equation*}
    \begin{cases}
    \gamma - \alpha > 1\\
    \gamma - 2\alpha > 1\\
    \gamma - \alpha - \beta < 0\\
    \gamma - \beta < 0
    \end{cases}
    \Longleftrightarrow
    \begin{cases}
    \gamma - \alpha > 1\\
    \gamma - \alpha - \beta < 0
    \end{cases}.
    \end{equation*}
    
    \item The cycle $\Gamma_{18}: O_2 \rightarrow O_{2, 3} \rightarrow O_3 \rightarrow O_{3, 4} \rightarrow O_{1, 3, 4} \rightarrow O_{1, 4} \rightarrow O_{1, 2, 4} \rightarrow O_{2, 4} \rightarrow O_2$ exists when the following condition is satisfied:
    \begin{equation*}
    \begin{cases}
    \gamma - \alpha > 1\\
    \gamma - 2\alpha > 1\\
    \gamma - \alpha - \beta < 0\\
    \gamma - \beta < 0\\
    \gamma - 2\beta < 0
    \end{cases}
    \Longleftrightarrow
    \begin{cases}
    \gamma - \alpha > 1\\
    \gamma - \alpha - \beta < 0
    \end{cases}.
    \end{equation*}
    
    \item The cycle $\Gamma_{19}: O_2 \rightarrow O_{2, 3} \rightarrow O_3 \rightarrow O_{3, 4} \rightarrow O_{1, 3, 4} \rightarrow O_{1, 2, 3, 4} \rightarrow O_{1, 2, 4} \rightarrow O_{2, 4} \rightarrow O_2$ exists when the following condition is satisfied:
    \begin{equation*}
    \begin{cases}
    \gamma - 2\alpha - \beta > 1\\
    \gamma - \alpha > 1\\
    \gamma - 2\alpha > 1\\
    \gamma - \alpha - \beta < 0\\
    \gamma - \alpha - 2\beta < 0\\
    \gamma - \beta < 0
    \end{cases}
    \Longleftrightarrow
    \begin{cases}
    \gamma - 2\alpha - \beta > 1\\
    \gamma - \alpha - \beta < 0
    \end{cases}.
    \end{equation*}
    
    \item The cycle $\Gamma_{20}: O_2 \rightarrow O_{2, 3} \rightarrow O_{1, 2, 3} \rightarrow O_{1, 2, 3, 4} \rightarrow O_{1, 2, 4} \rightarrow O_{2, 4} \rightarrow O_2$ exists when the following condition is satisfied:
    \begin{equation*}
    \begin{cases}
    \gamma - 2\alpha - \beta > 1\\
    \gamma - \alpha > 1\\
    \gamma - 2\alpha > 1\\
    \gamma - \alpha - \beta < 0\\
    \gamma - \alpha - 2\beta < 0\\
    \gamma - \beta < 0
    \end{cases}
    \Longleftrightarrow
    \begin{cases}
    \gamma - 2\alpha - \beta > 1\\
    \gamma - \alpha - \beta < 0
    \end{cases}.
    \end{equation*}
    
    \item The cycle $\Gamma_{21}: O_3 \rightarrow O_{1, 3} \rightarrow O_{1, 3, 4} \rightarrow O_{1, 2, 3, 4} \rightarrow O_{2, 3, 4} \rightarrow O_{2, 3} \rightarrow O_3$ exists when the following condition is satisfied:
    \begin{equation*}
    \begin{cases}
    \gamma - \alpha - \beta > 1\\
    \gamma - 2\alpha - \beta > 1\\
    \gamma - \alpha > 1\\
    \gamma - \alpha - 2\beta < 0\\
    \gamma - \beta < 0\\
    \gamma - 2\beta < 0
    \end{cases}
    \Longleftrightarrow
    \begin{cases}
    \gamma - \alpha - \beta > 1\\
    \gamma - \alpha - 2\beta < 0
    \end{cases}.
    \end{equation*}
\end{enumerate}

Taking into account the obtained inequalities, all heteroclinic cycles can be divided into the following groups:
\begin{enumerate}
    \item Set of cycles $\{\Gamma_1, \Gamma_5, \Gamma_6, \Gamma_{10}, \Gamma_{11}, \Gamma_{16}\}$ exists if the following condition is met:
    \begin{equation} \label{cond_hc_1}
    \alpha + 1 < \gamma < \beta.
    \end{equation}
    
    \item Set of cycles $\{\Gamma_2, \Gamma_9, \Gamma_{19}, \Gamma_{20}\}$ exists if the following condition is met:
    \begin{equation} \label{cond_hc_2}
    \begin{cases}
    \gamma - 2\alpha - \beta > 1\\
    \gamma - \alpha - \beta < 0
    \end{cases}.
    \end{equation}
    
    \item Set of cycles $\{\Gamma_3, \Gamma_{17}, \Gamma_{18}\}$ exists if the following condition is met:
    \begin{equation} \label{cond_hc_3}
    \begin{cases}
    \gamma - \alpha > 1\\
    \gamma - \alpha - \beta < 0
    \end{cases}.
    \end{equation}
    
    \item Set of cycles $\{\Gamma_4, \Gamma_7, \Gamma_{12}\}$ exists if the following condition is met:
    \begin{equation} \label{cond_hc_4}
    \begin{cases}
    \gamma - \alpha - \beta > 1\\
    \gamma - \beta < 0
    \end{cases}.
    \end{equation}
    
    \item Set of cycles $\{\Gamma_8, \Gamma_{14}, \Gamma_{15}, \Gamma_{21}\}$ exists if the following condition is met:
    \begin{equation} \label{cond_hc_5}
    \begin{cases}
    \gamma - \alpha - \beta > 1\\
    \gamma - \alpha - 2\beta < 0
    \end{cases}.
    \end{equation}
    
    \item Cycle $\Gamma_{13}$ exists if the following condition is met:
    \begin{equation} \label{cond_hc_6}
    \begin{cases}
    \gamma - 2\alpha - \beta > 1\\
    \gamma - \alpha - 2\beta < 0\\
    \gamma - \beta < 0
    \end{cases}.
    \end{equation}
\end{enumerate}

\section{Proof of proposition 1}
\begin{proposition} \label{prop1}
All heteroclinic cycles with equilibria at the vertices of the hypercube $0 \leq \rho_i \leq 1$ include at least six equilibria.
\end{proposition}
\begin{proof}
As noted above, all possible cycles contain an even number of equilibria. Therefore all heteroclinic cycles include at least four equilibria. Suppose there is a cycle $\Gamma$ containing four equilibria.

Since all edges with the same labels that are part of any cycle must go in the same direction, the only cycles that contain four equilibria are: $$\Gamma_1: O_1 \rightarrow O_{1, 3} \rightarrow O_{1, 3, 4} \rightarrow O_{1, 4} \rightarrow O_1,$$ $$\Gamma_2: O_4 \rightarrow O_{1, 4} \rightarrow O_{1, 2, 4} \rightarrow O_{2, 4} \rightarrow O_4,$$ $$\Gamma_3: O_{1, 4} \rightarrow O_{1, 2, 4} \rightarrow O_{1, 2, 3, 4} \rightarrow O_{1, 3, 4} \rightarrow O_{1, 4},$$ and the cycles $\Gamma'_1$, $\Gamma'_2$, and $\Gamma'_3$ obtained from the cycles $\Gamma_1$, $\Gamma_2$ and $\Gamma_3$ respectively by reversing the direction of the traversal.

As mentioned above, the edges labeled $E1$ and $E4$, as well as $E2$ and $E5$, have the same direction. Therefore, the cycles $\Gamma_1$, $\Gamma_2$, $\Gamma'_1$ and $\Gamma'_2$ do not exist.


For the cycle $\Gamma_3$ to exist, the following system of inequalities must be satisfied:
\begin{equation*}
\begin{cases}
\gamma - 2\alpha > 1\\
\gamma - \alpha - 2\beta > 1\\
\gamma - 2\beta < 0\\
\gamma - 2\alpha - \beta < 0
\end{cases}.
\end{equation*}
Subtracting the third inequality from the first, we obtain $\beta - \alpha > \frac{1}{2}$, subtracting the fourth inequality from the second, we obtain $\alpha - \beta > 1$. These inequalities are incompatible.

For the cycle $\Gamma'_3$ to exist, the following system of inequalities must be satisfied:
\begin{equation*}
\begin{cases}
\gamma - 2\beta > 1\\
\gamma - 2\alpha - \beta > 1\\
\gamma - 2\alpha < 0\\
\gamma - \alpha - 2\beta < 0
\end{cases}.
\end{equation*}
Subtracting the third inequality from the first, we obtain $\alpha - \beta > \frac{1}{2}$, subtracting the fourth inequality from the second, we obtain $\beta - \alpha > 1$. These inequalities are inconsistent. Therefore, the cycles $\Gamma_3$ and $\Gamma'_3$ also do not exist. Thus, all existing cycles contain at least six equilibria.
\end{proof}

\section{Proof of proposition 2}
\begin{proposition} \label{prop2}
All heteroclinic cycles with equilibria at the vertices of the hypercube $0 \leq \rho_i \leq 1$ contain one of the equilibria $O_1$, $O_2$, $O_3$ or $O_4$.
\end{proposition}
\begin{proof}
Suppose that there is a cycle $\Gamma$ that does not contain any of these equilibria. Then it contains at least one of the equilibria $O_{1, 2, 3}$, $O_{1, 2, 4}$, $O_{1, 3, 4}$, or $O_{2, 3, 4}$ (since there are no horizontal edges). Moreover, $\Gamma$ cannot contain equilibria $O_{1, 2}$, $O_{2, 3}$, and $O_{3, 4}$, since all edges directed downward from these equilibria have the same label.

Assume that $\Gamma$ does not contain the equilibrium $O_{1, 2, 3, 4}$. Then $\Gamma$ does not contain the equilibria $O_{1, 2, 3}$ and $O_{2, 3, 4}$, since all edges directed upward from these equilibria have the same labels. Therefore, $\Gamma$ also does not contain the equilibria $O_{1, 3}$ and $O_{2, 4}$, since there is only one edge from each of them to the remaining equilibria. It is easy to see that it is impossible to form a cycle from the remaining equilibria ($O_{1, 4}$, $O_{1, 2, 4}$ and $O_{1, 3, 4}$). This means that $\Gamma$ contains the equilibrium $O_{1, 2, 3, 4}$. The figure \ref{graph4_part} shows a graph with possible edges.

Since edges with the same labels have the same directions, the cycle $\Gamma$ can contain at most one equilibrium of the pair $O_{1, 2, 3}$ and $O_{2, 3, 4}$. Taking into account all the limitations, the existence of such cycles is possible: $$\Gamma_1: O_{1, 2, 3, 4} \rightarrow O_{1, 2, 3} \rightarrow O_{1, 3} \rightarrow O_{1, 3, 4} \rightarrow O_{1, 4} \rightarrow O_{1, 2, 4} \rightarrow O_{1, 2, 3, 4},$$ $$\Gamma_2: O_{1, 2, 3, 4} \rightarrow O_{1, 3, 4} \rightarrow O_{1, 4} \rightarrow O_{1, 2, 4} \rightarrow O_{1, 2, 3, 4},$$ $$\Gamma_3: O_{1, 2, 3, 4} \rightarrow O_{1, 3, 4} \rightarrow O_{1, 4} \rightarrow O_{1, 2, 4} \rightarrow O_{2, 4} \rightarrow O_{2, 3, 4} \rightarrow O_{1, 2, 3, 4},$$ as well as cycles $\Gamma'_1$, $\Gamma'_2$ and $\Gamma'_3$, obtained from $\Gamma_1$, $\Gamma_2$ and $\Gamma_3$, respectively, by replacing the direction of the traversal with the opposite one.
\begin{figure}[H]
    \centering
    \includegraphics[width = 1\linewidth] {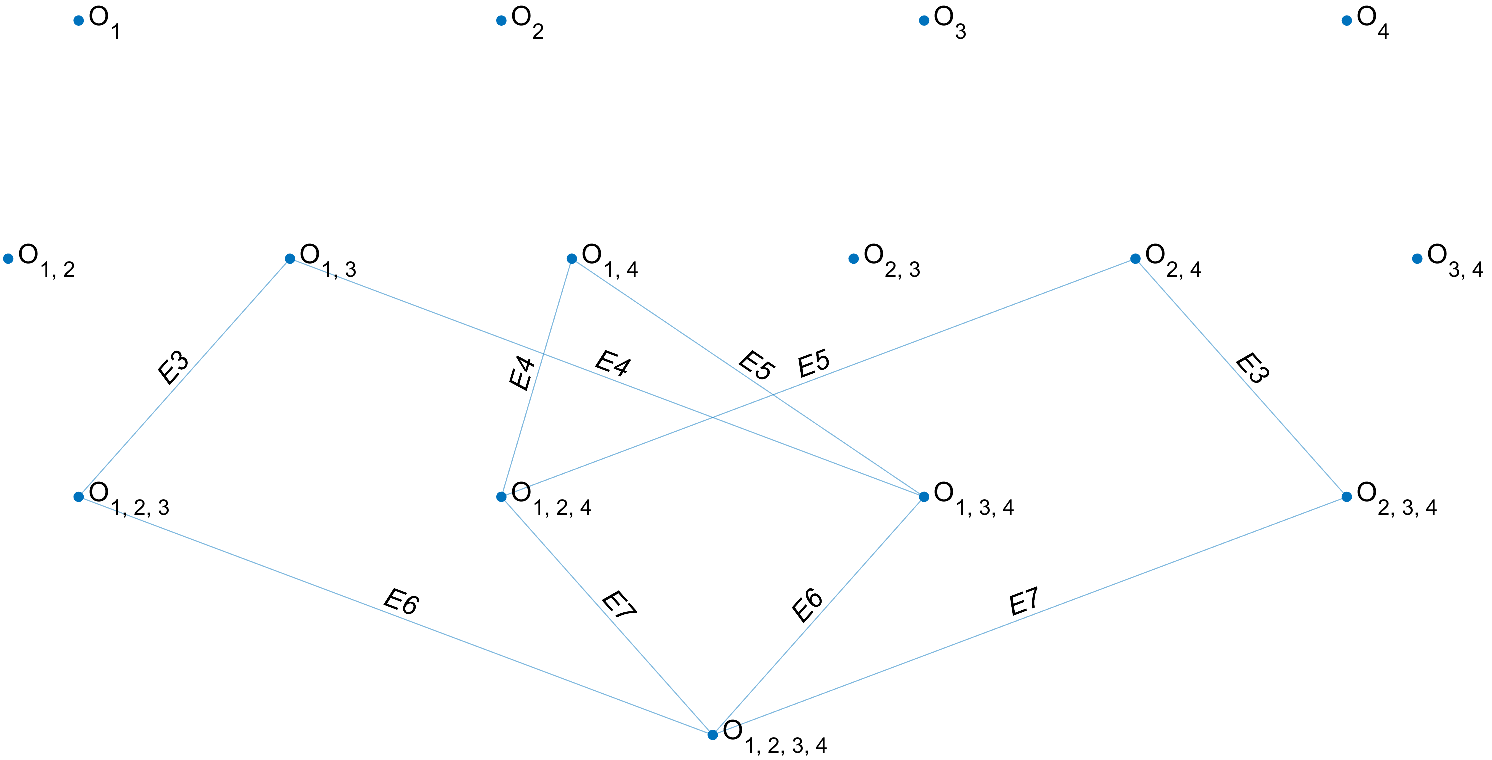}
    \caption{A graph containing all heteroclinic cycles that do not contain the equilibria $O_1$, $O_2$, $O_3$ and $O_4$.}
    \label{graph4_part}
\end{figure}

For the cycle $\Gamma_1$ to exist, the following conditions must be met:
\begin{equation*}
\begin{cases}
\gamma - \alpha - \beta > 1\\
\gamma - 2\alpha > 1\\
\gamma - \alpha - 2\beta > 1\\
\gamma - 2\beta < 0\\
\gamma - 2\alpha - \beta < 0
\end{cases}.
\end{equation*}
Subtracting the fourth inequality from the second one, we obtain the inequality $\beta - \alpha > \frac{1}{2}$. Subtracting the fifth inequality from the third one, we obtain the inequality $\alpha - \beta > 1$. These inequalities are incompatible.

For the cycle $\Gamma_2$ to exist, the following conditions must be met:
\begin{equation*}
\begin{cases}
\gamma - 2\alpha > 1\\
\gamma - \alpha - 2\beta > 1\\
\gamma - 2\beta < 0\\
\gamma - 2\alpha - \beta < 0
\end{cases}.
\end{equation*}
Subtracting the third inequality from the first one, and the fourth inequality from the second one, we obtain the same inequalities as in the previous case.

For the cycle $\Gamma_3$ to exist, the following conditions must be met:
\begin{equation*}
\begin{cases}
\gamma - 2\alpha > 1\\
\gamma - \alpha - 2\beta > 1\\
\gamma - \alpha - \beta < 0\\
\gamma - 2\beta < 0\\
\gamma - 2\alpha - \beta < 0
\end{cases}.
\end{equation*}
Subtracting the fourth inequality from the first one, and the fifth inequality from the second one, we obtain the same inequalities as in the first case.

Similarly, it can be shown that there are no cycles $\Gamma'_1$, $\Gamma'_2$, and $\Gamma'_3$. Thus, it is proven that all heteroclinic cycles contain at least one of the equilibria $O_1$, $O_2$, $O_3$ or $O_4$.
\end{proof}

\section{Proof of corollary 1}
\begin{corollary} \label{prop3}
A necessary condition for the existence of a heteroclinic cycle with equilibria at the vertices of the hypercube $0 \leq \rho_i \leq 1$ is the fulfillment of the inequality $\alpha + 1 < \gamma < \beta$ or $\beta + 1 < \gamma < \alpha$.
\end{corollary}
\begin{proof}
Since all heteroclinic cycles contain at least one of the equilibria $O_1$, $O_2$, $O_3$, or $O_4$ (see proposition \ref{prop2}), every heteroclinic cycle must contain edges labeled $E1$ and $E2$. Therefore, a necessary condition for the existence of cycles is the fulfillment of conditions \eqref{cond_1} and \eqref{cond_2}. Since the edges labeled $E1$ and $E2$ must be in different directions, then for heteroclinic cycles to exist, one of the following two inequalities must hold:
\begin{equation*}
\left[
\begin{gathered}
\alpha + 1 < \gamma < \beta\\
\beta + 1 < \gamma < \alpha
\end{gathered}
\right..
\end{equation*}
\end{proof}

\bibliographystyle{unsrt}
\bibliography{my_bib.bib}

\end{document}